\documentclass{article}   	
\usepackage{geometry}                		
\usepackage{graphicx}				
\usepackage{amssymb,latexsym}

\usepackage{amsmath}

\usepackage{amsthm}

\usepackage{multirow}

\usepackage{color}

\usepackage{url}

\newtheorem{lemma}{Lemma}

\title{New Public-Key Crypto-System EHT}
\author{Alessandro Budroni\quad Igor Semaev\\ 
\\Department of Informatics, University of Bergen, Norway\\ \\ \{alessandro.budroni,\,igor.semaev\}@uib.no}

\begin{document}
\maketitle
\begin{abstract}
In this note, an LWE  problem with a hidden trapdoor is introduced. It is used to construct an efficient   public-key crypto-system EHT. The new system is significantly different from LWE based NIST candidates like FrodoKEM. The performance of EHT compares favorably with FrodoKEM.  
\end{abstract}

\section{Introduction}
 The LWE (Learning with Errors) problem was introduced by Regev in \cite{Regev2005}, where  an LWE based  public-key  encryption was described. The problem was there proved to be hard  assuming the hardness of computing  shortest non-zero vectors  in general lattices. Since then 
several lattice based public-key  crypto-systems were invented, see  \cite{Peikert2016}. The NIST Post-Quantum Standardization Process stimulated interest for developing new quantum computer resistant public-key protocols. A number of submission to this competition are LWE or Ring LWE based, see \cite{nist}. In this note an LWE  problem with a hidden trapdoor is introduced. It is used to construct a new efficient   public-key crypto-system EHT.        
\section{ LWE problem}\label{Lwe-problem}
Let $n\leq m$ be positive integers and $q$ be a prime,  bounded by  a fixed polynomial in $n$ for large $n$.  
 Let $A$ be an integer $m\times n$-matrix of rank $n$ modulo  $q$ and $e=(e_1,\ldots,e_m)$ be a column  vector with entries  generated independently according to a non-uniform 
  distribution 
  \begin{equation}\label{distribution_00}
p[0],p[1],\ldots,p[q-1]
\end{equation}
on residues modulo $q$. Also let $x=(x_1,\ldots,x_n)$ be a column vector of integers and $Ax-e\equiv y\mod q$, where 
   $y=(y_1,\ldots,y_m)$. The problem is to find $x \mod q$ given $A,y,q$ and the distribution \eqref{distribution_00}.  The solution is unique for large enough $m>n$   depending on \eqref{distribution_00}. Commonly, a discrete normal (also called Gaussian) distribution $p_{\sigma}[a]$ with mean $0$ and variance $\approx\sigma^2$  is used. 
There are several ways to determine   $p_{\sigma}[a]$  on integers $a$. One may  define it as
\begin{equation}\label{Normal_Distribution}
p_{\sigma}[a]=\frac{e^{-a^2/2\sigma^2}}{\sum_{b\in\mathbb{Z}}e^{-b^2/2\sigma^2}}.
\end{equation}
 However, in practice, it is more convenient to generate a real $b$ according to an ordinary continuous normal distribution with mean $0$ and variance $\sigma^2$ and then round $b$ to the nearest integer $a$. 
 If the standard deviation $\sigma=o(q)$ as $q$ grows, one may assume that the distribution has its support on $a$   such that $|a|\leq (q-1)/2$. That defines a distribution $P_{\sigma}$  on the residues $0,1,\ldots,q-1$ modulo $q$.
 
 \section{Contributions}
Let $k$ be a positive integer. The public key of EHT is an $kn\times n$ matrix $A$ over residues modulo $q$, constructed in Section \ref{public_key}   and the distribution $P_{\sigma}$ introduced in Section \ref{Lwe-problem}.    
Let $x$ be a plain-text block. The cipher-text block $y$ is then computed as 
$y=Ax-e$, where the entries of $e$ are independently generated according to the distribution $P_{\sigma}$. In order to recover the plain-text block $x=(x_1,\ldots,x_n)$  given the cipher-text $y=(y_1,\ldots,y_{kn})$ one has to solve an LWE problem. The private key consists of three matrices $B,T,C$ described in Section \ref{private_key} such that $A=C^{-1}TB$ modulo $q$. With the private key the plain-text $x$  is recovered by a statistical procedure in Section \ref{decryption}.  The decryption failure  probability may be taken different by varying $k$, see Sections \ref{failure} and \ref{asymptotic_failure}. For instance, for using in KEM (Key Encapsulating Mechanisms) $k$ may be chosen to make the probability at most $10^{-5}$, while for PKE (Public Key Encryption) one can  make the probability $10^{-10}$.   With that choice the system is functional and more efficient, that is encryption/decryption work faster, than some of the NIST candidates with comparable security level. Explicit EHT parameters to fit three security levels are in Section \ref{parameters}. 
 The comparison with the NIST candidate FrodoKEM is in Section \ref{performance}. The decryption failure probability in EHT is   higher than that in the most submissions to the NIST competition. However, we have not found any efficient attacks that exploit the property.  There are two main approaches to the EHT cryptanalysis: find plain-text given cipher-text and find private key given public key. In Section \ref{cryptanalysis} the results of the cryptanalysis are presented. The asymptotic  complexity of breaking the crypto-system depends on the parameters, see Section \ref{EHT_parameters}, and is generally exponential in $n$ for large $n$.

The EHT crypto-system was invented by Semaev, who also  analysed  the decryption failure probability in Sections \ref{failure} and \ref{asymptotic_failure}, multiple encryptions of the same plain-text block in Section \ref{multiple_encryptions} and equivalent key recovering in Section \ref{basic_cryptanalysis}.
 The choice of EHT optimised parameters in Section \ref{parameters} and all computer experiments including those in  Section \ref{performance} are due to Budroni. Also, Budroni discovered   key-recovery attacks in Section \ref{variation_1}.    
  
\section{EHT Parameters and Encryption/Decryption } 
In this section the EHT crypto-system is described in detail.
\subsection{Parameters}\label{EHT_parameters}
Let $n,q,\sigma, k,\lambda,\gamma$ be positive integer parameters, where $q=n^{c_q},\sigma=n^{c_\sigma},\lambda=n^{c_\lambda}$ for some $c_q,c_\sigma,c_\lambda$, where $c_q>c_\sigma+c_\lambda$ and $k$ is to be defined later. 

\subsection{Private Key} 
\label{private_key}
The system private key consists of three matrices $B,T,C$.
\begin{enumerate}
\item The matrix $B$ is an integer $n\times n$ matrix of rank $n$ modulo $q$, whose rows are $B_1,B_2,\ldots,B_n$.
\item The matrix $T$ is an integer $kn\times n$ matrix  
$$T=
\begin{pmatrix}
t_{11}&0&\ldots&0\\
\ldots& &      &\\
t_{k1}&0&\ldots&0\\
0&t_{12}&\ldots&0\\
\ldots& &      &\\
0&t_{k2}&\ldots&0\\
\ldots& &      &\\
0&0&\ldots&t_{1n}\\
\ldots& &      &\\
0&0&\ldots&t_{kn}
\end{pmatrix} ,
$$
where the entries  $t_{ij}, 1\leq i\leq k, 1\leq j\leq n$ are non-zero modulo $q$. To reduce the decryption failure probability, in particular, to accept an incorrect plain-text block the following condition may be satisfied  for every $1\leq j\leq n$:
\begin{enumerate}
\item the residues $t_{ij}, 1\leq i\leq k$ are different,
\item for every  $a\ne 0\mod q$ not all the residues $at_{ij}\mod q, 1\leq i\leq k$ are close to $0$.
\end{enumerate}
The conditions are satisfied for randomly chosen $t_{ij}$ with high probability.

\item The matrix $C$ is an integer $kn\times kn$ matrix  whose rows $C_1,\ldots,C_{kn}$ are of norm at most $\lambda$. For a  correct decryption, the  dot-products $C_{u}C_{v}$ for $1\leq u<v\leq kn$   have to be very small, say bounded by  $\gamma$ in absolute value.
We  assume $C$ is invertible modulo $q$. The matrix $C$ may be constructed  with one of the following methods.
\end{enumerate}

\subsection{ Three Methods to Construct the Matrix $C$ }\label{matrix_C}
In this section we present  methods to construct the matrix $C$. 
 \begin{enumerate}
\item To define a row $C_s$  one may take a random subset $\{i_1,\ldots,i_{\lambda^2}\}$ of $\{1,2,\ldots,kn\}$ of size $\lambda^2$ and put
$$C_{sj}=
\begin{cases}  
\pm 1,& j\in \{i_1,\ldots,i_{\lambda^2}\},\\
0,& j\notin \{i_1,\ldots,i_{\lambda^2}\}.
\end{cases}
$$
The average of the dot-products $C_{u}C_{v}, u\ne v$ is $0$.
\item One can use combinatorial configurations(incidence structures). For instance, we may use 
a protective plane over a finite field of size $p$. The number of points on the plane is $p^2+p+1=nk$, there are $p^2+p+1$ lines on the plane, each line has $p+1=\lambda^2$ points, each two lines have exactly one point in common. The points on the plane and the lines may be permuted with  secret permutations $P$ and $Q$ accordingly.
We take lines as subsets $\{i_1,\ldots,i_{\lambda^2}\}$ and choose the rows of $C$ by inscribing $1$ or $-1$ to the points on the line and $0$ to all other points.  For this construction $\gamma=1$. 
 
\item Let $\lambda^2=2^s$ and $kn=r2^s$ for some integers $r,s$ and $H$ be a Hadamard matrix of size $2^s\times 2^s$. The rows and columns of $H$ are indexed by binary $s$-strings $a$ and $b$ respectively. The entry of $H$ in the row $a$ and the column $b$ is $(-1)^{ab}$, where $ab$ is the dot-product of $a,b$. Let $P,Q$ be  secret $nk\times nk$ permutation matrices. We set $C=P(H\otimes I_r) Q$, where $I_r$ is a unity matrix of size $r\times r$. That is 
$$C=P\begin{pmatrix}
H&0&\ldots&0\\
0&H&\ldots&0\\
& &\ldots &\\
0&0&\ldots&H
\end{pmatrix}Q,\qquad
C^{-1}=2^{-s}Q^{-1}\begin{pmatrix}
H&0&\ldots&0\\
0&H&\ldots&0\\
& &\ldots &\\
0&0&\ldots&H
\end{pmatrix}P^{-1}
$$ 
as $H^2=2^s\,I_{2^s}$.
The rows of $H$ are orthogonal, so the rows of $C$ are orthogonal too and $\gamma=0$.
In order to avoid weak keys found in Section \ref{variation_1}, instead of $P$ one uses a permutation $P^{*}$ defined in Section \ref{P*}.
\end{enumerate}

In all cases the rows of $C$ are  of norm $\lambda$. We will need to invert the matrix $C$ modulo $q$ to construct the system public key. For the first two methods the inversion is rather slow as the matrix size $kn\times kn$ may be large. For the third method the inversion is trivial. The experiments in Section \ref{performance} were conducted with EHT based on the third method for constructing $C$.  

 \subsection{Permutation $P^*$} 
\label{P*} The rows of the matrix $C$ may be split into chunks of size $k$ as $C_{1+k(i-1)},\ldots,C_{ki}$ for $i=1,\ldots,n$. Assume that two rows with the same positions for non-zero entries fall into one chunk. According to the cryptanalysis in Section \ref{variation_1}, they may be recovered faster than in the general case and this may lead to a weak key of the crypto-system. For instance, let $n=8,k=2,\lambda^2=4$ and $P$ be defined by Table \ref{Table_P}, and $C=P(H\otimes I_r) Q$ as above.
\begin{table}
\caption{Permutation $P$}
\begin{center}
\begin{tabular}{|c|rrrrrrrrrrrrrrrr|}\hline
$x$&1&2&3&4&5&6&7&8&9&10&11&12&13&14&15&16\\\hline
$T(x)$&14&1&10&9&6&13&5&16&15&2&12&3&7&11&8&4\\\hline
\end{tabular}
\end{center}
\label{Table_P}
\end{table}%
The pairs of rows $(C_5,C_6)$, $(C_7,C_8)$ and $(C_9,C_{10})$ have the same non-zero positions. One needs to choose   $P=P^{*}$ to avoid this. That is easy to do if $\lambda^2|n$.

One first defines a mapping $M:\{1,2,\ldots,nk\}\rightarrow\{1,2,\ldots,n\}$ such that the values $M(l+(j-1)\lambda^2)$ are different for $l=1,\ldots,\lambda^2$ and every fixed $j$ from $1\leq j\leq nk/\lambda^2$. Then one defines a permutation $P^{*}$ on $1,2,\ldots,nk$ by the rule
$P^{*}(l+(j-1)\lambda^2)\in\{1+k(i-1),\ldots,ki\}$, where $i= M(l+(j-1)\lambda^2)$. For instance, the mapping  $M$ and the permutation $P^{*}$ may  be defined by Table \ref{Table_M}.
\begin{table}
\caption{Mapping $M$ and Permutation $P^{*}$ }
\begin{center}
\begin{tabular}{|c|rrrrrrrrrrrrrrrr|}\hline
$x$&1&2&3&4&5&6&7&8&9&10&11&12&13&14&15&16\\\hline
$M(x)$&1&2&3&4&5&6&7&8&1&3&5&7&2&4&6&8\\\hline
$P^{*}(x)$&1&3&5&7&9&11&13&15&2&6&10&14&4&8&12&16\\\hline
\end{tabular}
\end{center}
\label{Table_M}
\end{table}%
The number of possible mappings $M$ and therefore the number of possible permutations $P^{*}$ is very large. So this restriction does not  affect the security.

\subsection{Public Key}\label{public_key} The public key is an integer $kn\times n$ matrix $A\equiv C^{-1}TB\mod q$ and $\sigma$.

\subsection{Encryption} Let $x$ be a plain-text encoded by a column vector of size $n$ of residues modulo $q$. We may assume  the plain-text has some natural redundancy. That significantly reduces the probability of accepting an incorrect plain-text block $x$ as correct at the decryption.  If the plain-text does not have  natural redundancy then it may be artificially introduced.  For the parameters chosen in Section \ref{parameters} we need at most $2\log_2q$ bits redundancy, which
 may be provided with a linear $\mathbb{F}_q$-code of length $n$ and  dimension $n-2$. This only negligibly affects the security of the crypto-system. Another option is to use Optimal Asymmetric Encryption Padding (OAEP), see \cite{BR94}. 
 
  The cipher-text is 
$y\equiv Ax-e\mod q$, where $e$ is a column vector of size $m$ whose entries are residues modulo $q$ independently generated according to the distribution $P_{\sigma}$ defined by \eqref{Normal_Distribution}.   
\subsection{Decryption}\label{decryption} By multiplying the both sides of $y\equiv Ax-e\mod q$ with $C$ one gets
$z\equiv Tb-Ce\mod q,$
 where  
$z=Cy=(z_1,\ldots,z_{kn}),b=Bx=(b_1,\ldots,b_{n})$
are column vectors of length $kn$ and $n$ respectively.  The entries of  $Ce$ have  discrete normal distributions $ P_{\sigma\lambda}=(p_{\sigma\lambda}[0],\ldots,p_{\sigma\lambda}[q-1])$ with mean $0$ and standard deviation  $\approx\sigma\lambda=o(q)$, as $c_q>c_\sigma+c_\lambda$. By the definition of the matrix $T$ in Section \ref{private_key}, 
\begin{eqnarray}
&C_{1+k(i-1)}e\equiv t_{1i}b_i-z_{1+k(i-1)} \mod q,\nonumber\\
&\ldots\nonumber\\
&C_{k+k(i-1)}e\equiv t_{ki}b_i-z_{k+k(i-1)}\mod q.\nonumber
\end{eqnarray}
One now finds $b_i=B_ix, 1\leq i\leq n $ with a statistical procedure. For each residue $a$ 
 the value of the statistic
\begin{equation}\label{statistic}
S_i(a)=\sum_{j=1}^k\ln(q\,p_{\sigma\lambda}[t_{ji}a-z_{j+k(i-1)}])
\end{equation}
is computed, where $t_{ij}a-z_{j+k(i-1)}$ is taken modulo $q$.  One decides $b_i=a$ if $ S_i(a)>0$. In practice, it may be more convenient to
 use an ordinary continuous normal density  $p_{\sigma\lambda}[x]=e^{-x^2/2(\sigma\lambda)^2}/\sqrt{2\pi(\sigma\lambda)^2}$ for  residues $x$ modulo $q$ such that $|x|\leq(q-1)/2$. Also, one may decide $b_i=a$ for a residue $a$ which maximises $S_i(a)$, however, in this case the decryption failure probability is   difficult to predict. By a decryption failure we mean rejecting the correct plain-text block $x$ or accepting an incorrect one. 
If $k$ is large enough, then the decryption is unique and correct with high probability as it follows from the next Section.

\subsection{Decryption Failure Probability}\label{failure}
The asymptotic analysis in Section \ref{asymptotic_failure} shows that $k$ may be taken constant for large $n$ if $c_{\sigma}+c_{\lambda}<c_{q}$. In this section we find explicit expressions for the decryption failure probability.  The distribution $P_{\sigma\lambda}$ may be approximated by an ordinary continuous normal density with mean $0$ and standard deviation $\sigma\lambda$, so
\begin{equation}\label{approx} 
\ln(q\,p_{\sigma\lambda}[x])\approx\ln(\frac{q}{\sigma\lambda\sqrt{2\pi}})-\frac{x^2}{2(\sigma\lambda)^2}.
\end{equation}
Then $S_i(a)\approx\sum_{j=1}^k \ln(\frac{q}{\sigma\lambda\sqrt{2\pi}})-\frac{x_j^2}{2(\sigma\lambda)^2}$, where $x_j\equiv t_{ji}a-z_{j+k(i-1)}\mod q$.
 Let $\delta^2=2k(\sigma\lambda)^2\ln(\frac{q}{\sigma\lambda\sqrt{2\pi}})$. According to the test $S_i(a)>0$, one decides $b_i=a$ if 
 \begin{equation}\label{ball}
\sum_{j=1}^k x_j^2<\delta^2.
\end{equation}
We consider two cases. First, let $b_i=a$. Then $x_1,\ldots,x_k$ are independently generated with the distribution $P_{\sigma\lambda}$. The success (that is to accept  $b_i=a$)  probability is  
$$\beta=\mathbf{Pr}(S_i(a)>0)\approx\mathbf{Pr}(\chi^2_k<(\delta/\sigma\lambda)^2),$$
where $\chi^2_k$ denotes a random variable distributed as $\chi$-square with $k$ degrees of freedom. The probability to reconstruct correctly all $b_i=B_ix$ and therefore to recover the correct $x$  is then $\beta_1=\beta^n$.
 Second, let $ b_i\ne a$. Then 
 $$x_j\equiv t_{ji}a-z_{j+k(i-1)}\equiv
 t_{ji}(a-b_i)+C_{j+k(i-1)}e \mod q.$$
So  
  we can assume $x_1,\ldots,x_k$ are independently and uniformly distributed over residues modulo $q$. By \eqref{ball}, the probability to decide $b_i=a$ for an incorrect $a$ is
$$\alpha\approx \frac{\hbox{Vol}_k(\delta)}{q^k}=\frac{\pi^{k/2}\,\delta^k}{\Gamma(k/2+1)\,q^k},$$
where $\hbox{Vol}_k(\delta)$ is the volume of a $k$-ball of radius $\delta$. Therefore, the probability to accept at least one incorrect $b_i=a$ is at most $\alpha_1=nq\alpha$. If the plain-text block $x$ is a codeword in a linear $\mathbb{F}_q$-code of length $n$ and dimension $n-2$ the probability to accept an incorrect plain-text block is at most $\alpha_1q^{-2}=n\alpha/q$. Therefore, the decryption failure (to reject the correct $x$ or to accept an incorrect one)  probability is at most $(1-\beta_1)+\alpha_1q^{-2}\approx (1-\beta_1)$ for the parameters in Section \ref{parameters}. So the most probable failure is to reject both the correct and incorrect plain-text blocks $x$ and the probability of this is around $ (1-\beta_1)$. Obviously, the probability to  reject the correct $x$ and to accept an incorrect one is much smaller.

 Let $q=1021,n=128,\sigma=5.105,\lambda=4$. The following Figure \ref{alpha_beta}  shows the experimental probabilities $1-\beta_1$ and $\alpha_1$ of the decryption failures for a variety of $k=7,\ldots,13$ in comparison with their estimates. The estimates are  rather accurate upper bounds for the actual probabilities.

\begin{figure}[h]
  \center{\includegraphics[width=15cm]{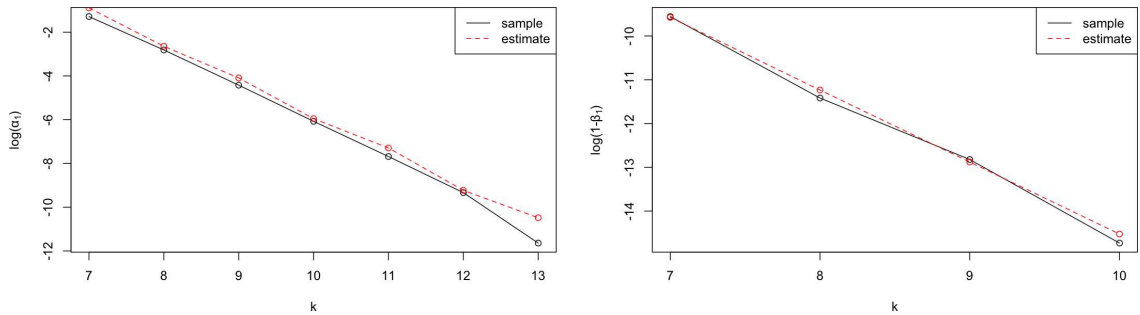}}\\
  \caption{Experimental and estimated error probabilities $\alpha_1$ and $1-\beta_1$ }\label{alpha_beta}
\end{figure}

 \subsection{Complexity Parameters}
 The size of the public key is $kn^2$ residues modulo $q$. The  encryption cost (to encrypt $n\log_2q$ bits of the plain-text) is $ kn^2$ multiplications modulo $q$. For binary $x$ the cost (to encrypt $n$ bits of the plain-text) is $ kn^2$ additions modulo $q$. The decryption cost is $kn\lambda^2$ additions modulo $q$ to compute $z=Cy$. This may be improved to $kn\log_2(\lambda^2)$ additions with the fast Walsh-Hadamard transform. Then it takes  
  $knq$ additions modulo $q$  to compute $nq$ values of the statistic, and $n^2$ operations to solve the linear equation system $b_i=B_ix, i=1,\ldots,n$ to recover $x$, where the inversion of $B$ modulo $q$ is precomputed. 
  The matrix $B$ may be chosen to accelerate the solution of the system  to $<n^2$ operations.  
\subsection{Variant}\label{variant}
To construct private/public key one may take  
$$T=
\begin{pmatrix}
t_{11}&0&\ldots&0\\
\ldots& &      &\\
t_{k1}&0&\ldots&0\\
t_{k+1\,1}&t_{1\,2}&\ldots&0\\
\ldots& &      &\\
t_{2k\,1}&t_{k\,2}&\ldots&0\\
\ldots& &      &\\
t_{(n-1)k+1\,1}&t_{(n-2)k+1\,2}&\ldots&t_{1\,n}\\
\ldots& &      &\\
t_{nk\,1}&t_{(n-1)k\,2}&\ldots&t_{k\,n}
\end{pmatrix}. 
$$ 
 The decryption algorithm is then easy to adjust.
\section{Cryptanalysis}\label{cryptanalysis} We consider two approaches to the cryptanalysis of LWE-Trapdoor crypto-system:  find plain-text from the cipher-text and recover private key from the  public key.
\subsection{ Plain-Text Recovering}\label{pl_text_recover}
 To reconstruct  the plain-text $x$ from the cipher-text $y$ one has to solve an LWE problem $y\equiv Ax-e\mod q$, where $A$ is a matrix of size $kn\times n$. According to \cite{HKM17} the complexity of this problem is exponential $2^{O(n)}$ with both lattice-base algorithms and amplified BKW if the structure of $A$ is not taken into account.

  Let's consider the so called primal attack.
 One may  assume that $A=\begin{pmatrix}A_1\\A_2\end{pmatrix}$, where $A_1$ is a matrix of size $n\times n$,  invertible modulo $q$ and the matrix $A_2$ is of size $(kn-n)\times n$, and let $y=\begin{pmatrix}y_1\\y_2\end{pmatrix}$ and $e=\begin{pmatrix}e_1\\e_2\end{pmatrix}$ accordingly. Then  $A_1x-e_1=y_1,A_2x-e_2=y_2$. So $x=A_1^{-1}e_1+y_1$ and $A_2A_1^{-1}e_1-e_2-y_2+A_2A_1^{-1}y_1=0$. The vector $w=(e,1)$ satisfies $A'w\equiv 0 \mod q$, where $A'$ is a horizontal concatenation of
 $A_2A_1^{-1},-I_{kn-n},-y_2+A_2A_1^{-1}y_1$. The size of $A'$ is  $(kn-n)\times (kn+1)$. 
 
  Let $m\leq kn$ be a parameter and  
 let $B$ be a matrix constructed with $m$ rows of  $A'$ ignoring some zero columns. The size of $B$ is $m\times (m+n+1)$. Then 
$A'w\equiv 0 \mod q$ implies $Bv\equiv 0 \mod q$, where $v=(f,1)$ and $f$ is a sub-vector of $e$ of length $m+n$.  
 Thus  $v$ belongs to a lattice of dimension $m+n+1$ and of volume $\hbox{Vol}=q^{m}$.
 
  The expected norm of $v$ is at most $\sqrt{\sigma^2 (m+n)+1}\approx \sigma\sqrt{m+n}$.  The Block Korkine-Zolotarev (BKZ) algorithm is applied to reduce the basis of the lattice and recover the shortest non-zero vector, which is likely to be $v$. Then $x$ is computed by solving a system of linear equations modulo $q$. The reduction algorithm calls   Shortest Vector Problem (SVP) oracle for a lattice of a smaller dimension $\beta\leq m$  a polynomial number of times \cite{Hanrot}. The cost of one call is $\geq 2^{0.292\beta}$ operations with a sieving algorithm according to \cite{Laarhoven}. One takes the smallest $\beta$   to satisfy the inequalities
$$ \sqrt{\beta}\sigma\leq \delta^{2\beta-d}\,\hbox{Vol}^{1/d},\quad \delta=\left((\pi\beta)^{1/\beta}\beta/2\pi e\right)^{1/2(\beta-1)}>1.$$

The sieve algorithms may benefit from the Grover's quantum search and this pushes the complexity down to $2^{0.265b}$ \cite{LaarhovenMosca,Laarhoven15}. As remarked in \cite{AD}, since these algorithms require building a list of lattice vectors of size of $2^{0.2075b}$, it is plausible to believe that the best quantum SVP algorithm would run in $2^{0.2075b}$.
    
The so called dual attack \cite{Regev2005} is based on finding a large number of short vectors $u,v$ such that $uA\equiv v\mod q$. One then applies a statistical test to recover the entries of $x$. The dual attack is generally inferior compared to the primal attack, see \cite{FrodoKEM}.

\subsection{Multiple Encryptions}\label{multiple_encryptions}
Assume that the same plain-text block $x$ was encrypted  $s$ times with independently generated error vectors $e(i)=(e_{i1},\ldots,e_{im}), i=1,\ldots,s$. So $s$ cipher-texts $y(i)=Ax-e(i)$ are available for the cryptanalysis, where $y(i)=(y_{i1},\ldots,y_{im})$ and $|y_{ij}|<q/2$.  Let's fix an index $j$ in $1,\ldots,m$ and let $A_j$ be a row in $A$. Then $y_{ij}=A_jx-e_{ij}\mod q$  are taken from a discrete normal distribution with mean $A_jx \mod q$ and variance $\sigma^2$.   So $\sum_{i=1}^s y_{ij}/s$ is  normally distributed  with  mean $A_jx \mod q$ and variance $\sigma^2/s$. If $s=O(\sigma^2)$, then   one recovers $A_jx \mod q$ as $\sum_{i=1}^s y_{ij}/s\approx A_jx \mod q$. One does this for  some $n$ values of $j$ and then solves a system of linear equations to recover $x$. 

A kind of DOS (denial of service) attack is applicable to recover $x$ with this method. To protect against that sort of attacks,  each plain-text block may be made dependent on a counter value or time stamp before encryption.

\subsection{Equivalent and Private Key Recovering} 

\subsubsection{Framework}
 Let   $C,T,B$ be any triplet satisfying the restrictions of Section \ref{private_key} and $CA=TB$. Let the norm of the rows of $C$ is at most $\mu$. If $\mu$ is not significantly larger than $\lambda$, then  $C,T,B$ is the system equivalent private key and it may be used to decrypt data. However, if $\mu$ is significantly larger than $\lambda$, then the decryption fails as it  is not  unique. 
The equality $CA=TB$ implies
 $$
C_{1+k(i-1)}A=t_{1i}B_i,
\ldots,
C_{k+k(i-1)}A=t_{ki}B_i,$$
for $i=1,\ldots,n$ and the rows of $C$ and $B$. Let's assume $t_{1i}=1$ and  eliminate $B_i$ from the equations (if $t_{1i}\ne1$, then we eliminate $t_{1i}B_i$).   The vector $\bar{C}_i=(C_{1+k(i-1)},\ldots,C_{k+k(i-1)})$ is of size $k^2n$ and it satisfies
\begin{equation}\label{CA=TB}
\bar{C}_i\bar{A}\equiv 0\mod q,\quad\bar{A}=
\begin{pmatrix}
t_{2i}A&t_{3i}A&\ldots&t_{ki}A\\
-A&0&\ldots&0\\
0&-A&\ldots&0\\
\ldots& & &\\
0&0&\ldots&-A
\end{pmatrix},
\end{equation}
where $\bar{A}$ is a matrix of size $k^2n\times (k-1)n$.      
Therefore, $\bar{C}_i$ belongs to a lattice of  volume $q^{n(k-1)}$ and of dimension $k^2n$ for some residues $t_{2i},\ldots,t_{ki}$.
The norm of  $\bar{C}_i$ is at most $\mu\sqrt{k}$.

 \subsubsection{ Private Key Recovering}\label{variation_1}
   The rows of $C$ have to satisfy \eqref{CA=TB} and so   
\begin{eqnarray}\label{equality_2}
(
C_{1+k(i-1)},C_{j+k(i-1)})
\begin{pmatrix}
t_{ji}A\\
-A
\end{pmatrix}\equiv 0\mod q.
\end{eqnarray}
 One has to  find all pairs of non-collinear vectors $C_1,C_2$ of length $kn$ with $\lambda^2$ non-zero entries $\pm1$ each such that
 $ tC_1A-C_2A\equiv 0\mod q$ for some residues $t$. Then $C$ is easy to recover and the matrices $T$ and $B$ are  easy to recover too.
 
\begin{description}
\item[Brute Force Attack.]One  brute forces such $C_1,C_2$ and finds  collinear rows (matches) $C_1A,C_2A$. To this end one divides (scales) $C_1A,C_2A$ by their left most non-zero entries and searches for all the matches. The  
  number of trials is  $2^{2\lambda^2}{kn\choose \lambda^2}{kn-\lambda^2\choose \lambda^2}$. 
    \item[Time-Memory-Trade-Off.] One defines the table of all scaled $C_1A$. The table is sorted to find all the matches. 
 At the expense of  memory size $V=2^{\lambda^2}{kn\choose \lambda^2}$, the complexity is $O(V\log V)$.
 \item[Lattice Based Method.]
 Let $h\leq kn$ be an integer parameter and $S_1,S_2\subseteq\{1,2,\ldots,kn\}$ be two subsets of size $h$. Let $A_i$ be a sub-matrix constructed with  $h$ rows of  $A$ whose indices are in $S_i$. If the positions of non-zero entries of $C_i$ are in $S_i$, then $ tC_1A_1-C_2A_2\equiv 0\mod q$. Therefore $(C_1,C_2)$ belongs to a lattice of dimension $2h$ and of  volume $q^n$ for a residue $t$ modulo $q$. The norm of $(C_1,C_2)$ is $\lambda\sqrt{2}$.  For each  $t$ one tries random subsets $S_1,S_2$ and applies BKZ algorithm until  $(C_1,C_2)$ is found.   The number of BKZ applications  to recover all the rows of $C$ is
 $q\left({kn\choose{h}}/{kn-\lambda^2\choose{h-\lambda^2}}\right)^2$        
 on the average. One finds $h$ to minimise the overall cost.
\item[Possible Weak Keys.] Assume that the non-zero positions of  $C_1=C_{j_1+k(i-1)}$ and $C_2=C_{j_2+k(i-1)}$ are the same for some $1\leq j_1<j_2\leq k$ and  $1\leq i\leq n$.  This event  happens quite often with the third method for constructing $C$ in Section \ref{matrix_C} when a random row permutation $P$ is used and  $\lambda^2>\sqrt{n}$.  To recover $C_1$ and $C_2$    one  tries only one  subset $S\subseteq\{1,2,\ldots,kn\}$ of size $h$ not two. One applies BKZ 
 to a lattice of dimension $2h$ and volume $q^n$ to recover  $C_{1},C_{2}$. 

The vector $C_1A$ is collinear to the row   $B_i$  of  $B$. So  $C_{j+k(i-1)}A\in \langle B_i\rangle=\langle C_1A\rangle$ for every $1\leq j\leq k$. Let the non-zero positions of $C_{j+k(i-1)}$ be in a subset $S_1\subseteq\{1,2,\ldots,kn\}$ of size $h_1\geq \lambda^2$ and let $A_1$ be a sub-matrix of $A$ constructed with $h_1$ rows of $A$ whose indices are in $S_1$. Then $C_{j+k(i-1)}A_1\in \langle B_i\rangle$. The vector $C_{j+k(i-1)}$ of norm $\lambda$ belongs to a lattice of dimension $h_1$ and of volume $q^{n-1}$. One guesses the subset  $S_1$ and runs BKZ to recover $C_{j+k(i-1)}$. With a similar approach one finds all the rows of $C$. 
 The values $h,h_1$ are chosen to minimise the overall cost. The method is significantly faster than all the previous methods of cryptanalysis. However, the event  does not happen for any $i$   if the row permutation $P^{*}$ instead of $P$ is used to construct $C$, see Section \ref{P*}. The system does not then admit such weak keys.     

    \end{description}

\subsubsection{  Equivalent Key Recovering}
 \label{basic_cryptanalysis}
 To construct an equivalent key one may try to  find short vectors $\bar{C}_i$(say of norm $\leq \mu\sqrt{k}$) such that  $\bar{C}_i\bar{A}\equiv 0\mod q$, see \eqref{CA=TB}. Thus $\bar{C}_i$ belongs to a lattice of dimension $k^2n$ and volume $q^{(k-1)n}$. 
  Then   $B_i\equiv C_{1+k(i-1)}A$ thus constructing the system equivalent  key. Since the size of $\bar{A}$ is very large, that seems a hard task.

For some  $t_{2i},\ldots,t_{ki}$ the lattice may contain   very short vectors $\bar{C}_i$ which do not depend on $A$, see Section \ref{variation_2} below.  The component vectors $C_{1+k(i-1)},\ldots,C_{k+k(i-1)}$ are then collinear. The random variables $C_{1+k(i-1)}e,\ldots,C_{k+k(i-1)}e$ are dependent and the   decryption fails.

 \subsubsection{ Very short vectors in the lattice $\bar{C}\bar{A}\equiv 0 \mod q$}\label{variation_2}
 Let
 $$\bar{I}=
\begin{pmatrix}
t_{2}I&t_{3}I&\ldots&t_{k}I\\
-I&0&\ldots&0\\
0&-I&\ldots&0\\
\ldots& & &\\
0&0&\ldots&-I
\end{pmatrix},
$$  
 for some non-zero residues $t_2,\ldots,t_k$, where $I$ is a unity matrix of size $kn\times kn$ and so $\bar{I}$ is a matrix of size $k^2n\times k(k-1)n$.  Let
$$\bar{R}=
\begin{pmatrix}
A&0&\ldots&0\\
0&A&\ldots&0\\
\ldots& & &\\
0&0&\ldots&A
\end{pmatrix},
$$  
be a matrix of size $k(k-1)n\times (k-1)n$. Obviously, $\bar{A}=\bar{I}\bar{R}$ for $t_j=t_{ji}$. Let an integer vector 
$$\bar{C}=(c_1,c_2,\ldots,c_{k^2n})$$ of size $k^2n$ satisfy 
$\bar{C} \bar{I}\equiv 0 \mod q$, then $\bar{C} \bar{A}\equiv 0 \mod q$. We want to  construct  $\bar{C}$ with a low norm. The congruence  $\bar{C} \bar{I}\equiv 0 \mod q$ implies
\begin{equation}\label{lattice_L}
\begin{pmatrix}
c_{i},c_{i+kn},\ldots,c_{i+k(k-1)n}
\end{pmatrix}
\begin{pmatrix}
t_{2}&t_{3}&\ldots&t_{k}\\
-1&0&\ldots&0\\
0&-1&\ldots&0\\
\ldots& & &\\
0&0&\ldots&-1
\end{pmatrix}\equiv 0\mod q
\end{equation}
for $i=1,\ldots,kn$. 
Let
$V_i=(c_{i},c_{i+kn},\ldots,c_{i+k(k-1)n})$. 
 The vectors $V_i$ belong to the same lattice $L(t_2,\ldots,t_k)$  of dimension $k$ and volume $q^{k-1}$ defined by \eqref{lattice_L}. For the majority of $t_2,\ldots,t_k$ the shortest non-zero vector of the lattice  has norm around $\sqrt{k}\,q^{1-1/k}$. Those vectors may be used to construct a vector $\bar{C}$   with norm around $\sqrt{k^2n}\,q^{1-1/k}$. That is much larger than $\lambda\sqrt{k}$, the norm of the hidden vectors $(C_{1+(i-1)k},\ldots,C_{k+(i-1)k})$ used in the construction of the public matrix $A$. However,  the lattice $L(t_2,\ldots,t_k)$  contains the vector $(1,t_2,\ldots,t_k)$ which is  fairly short for  small in absolute value  non-zero residues $t_2,\ldots,t_k$. For instance, for $t_i=\pm1$ the lattice  contains a vector of norm $\sqrt{k}$.  The following lemma shows that all very short vectors in the lattice are collinear.   
\begin{lemma}\label{lattice_1}
Let $(v_{1},v_{2},\ldots,v_{k})$ be a non-zero vector in $L(t_2,\ldots,t_k)$, where $\gcd(v_1,\ldots,v_k)=1$  and $|v_i|<\sqrt{q/2}$ and let $(w_{1},w_{2},\ldots,w_{k})$ be another vector there such that $|w_i|<\sqrt{q/2}$. Then
$$(w_{1},w_{2},\ldots,w_{k})=a(v_{1},v_{2},\ldots,v_{k})$$
for an integer $a$.
\end{lemma}

\begin{proof}
As $v_i=t_iv_1$ and $w_i=t_iw_1$, we get $v_1w_i- v_iw_1\equiv 0\mod q$ for $i=2,\ldots,k$. Since, $|v_1w_i- v_iw_1|<q$ we have $v_1w_i= v_iw_1$. Therefore,
$$(w_{1},w_{2},\ldots,w_{k})=\frac{w_1}{v_1}\,(v_{1},v_{2},\ldots,v_{k})$$
As $v_i$ are coprime, that is only possible when $v_1|w_1$. That proves the lemma. 
\end{proof}
So if the entries of every $V_i$ are $<\sqrt{q/2}$, then
$V_i=a_i(v_{1},v_{2},\ldots,v_{k})$
for the same vector $(v_{1},v_{2},\ldots,v_{k})$ and integer $a_i$.
 Then the vector $\bar{C}=(C_1,\ldots,C_k)$  satisfies $C_i=v_i(a_1,a_2,\ldots,a_{kn})$. Every entry of $C_i$ is $<\sqrt{q/2}$ and may be fairly small, for instance, when $t_i=\pm1$. In the latter case the norm of $\bar{C}$ is $\sqrt{k^2n}$. 
 
  However, since $C_{1},\ldots,C_{k}$ are 
collinear, the random variables $C_{j}e$ are  dependent. The statistic \eqref{statistic} does not provide proper decryption for $b_i=B_ix$.

We may admit some few of the vectors $V_i$ with entries $\geq \sqrt{q/2}$ and the rest of the vectors $V_i$ with entries $<\sqrt{q/2}$.   Then the norms of  $C_{1},\ldots,C_{k}$ become significantly larger and  there still is  a significant statistical dependence  between $C_{j}e$. So the decryption is not unique in this case either.
\section{Proposed Parameters} \label{parameters}  
In Table \ref{table:proposedParameter} we propose three sets of parameters matching Levels 1,\,3,\,5 in the call  of the NIST Post-Quantum Cryptography Standardization Process \cite{nist}. These parameters correspond to or exceed the brute-force security of AES-128,\,192,\,256  respectively.
The parameters are chosen to 
 minimise the encryption + decryption cost subject to the following constraints:  the security 
 provided according to the most efficient attack (primal attack for the underlying LWE problem, see Section \ref{pl_text_recover})  
 is close to the targeted security level,
 for implementation efficiency we set $q<2^{16}$, where $q$ is  close to  a power of 2, and $n$ is a multiple of 64, and
 $\lambda^2$ is a power of 2, where $\lambda^2|n$.
Two sets of parameters for each security level reaching a decryption failure probability of at most $10^{-5}$(type $A$) and $10^{-10}$(type $B$) respectively are provided. The decryption failure probability $1-\beta_1$ is computed according to Section \ref{failure}. 

\begin{table}[h]
\renewcommand{\arraystretch}{1.3}
\begin{center}
\begin{tabular}{|c|c|cccccc|}
\hline
Name & Level & $ n $ & $ k $ & $ q $ & $ \lambda^2 $ & $ \sigma $ & $1-\beta_1$ \\ \hline
EHT-light-A & \multirow{2}{*}{1} & $ 256  $ & $ 16 $ & $ 1021 $ & $ 32 $ & $ 8.8 $ & $7.4\cdot 10^{-6}$ \\ 
EHT-light-B &  & $ 256 $ & $ 25 $ & $ 2039 $ & $ 32 $ & $ 14.5 $ & $4.6\cdot 10^{-11}$ \\ \hline
EHT-medium-A & \multirow{2}{*}{3} & $ 384 $ & $ 14 $ & $ 2039 $ & $ 32 $ & $ 13.5  $ & $3.3\cdot 10^{-6}$ \\ 
EHT-medium-B &  & $ 384 $ & $ 24 $ & $ 2039 $ & $32$ & $ 13.5  $ & $4.8\cdot 10^{-11}$ \\ \hline
EHT-high-A & \multirow{2}{*}{5} & $448$ & $ 17 $ & $ 2039 $ & $32$ & $ 17.5 $ & $5.2\cdot 10^{-6}$ \\ 
EHT-high-B &  & $448$ & $ 24 $ & $ 4091 $ & $32$ & $ 27.0 $ & $5.6\cdot 10^{-11}$ \\ \hline
\end{tabular}
\end{center}
\caption{Proposed parameters.}
\label{table:proposedParameter}
\end{table}

The sizes of the public key $A$, the private key $C, T, B$, the plain-text $x$ and the cipher-
text $y$ for each security level are presented in Table \ref{table:sizeKeys}.

\begin{table}[h]
\renewcommand{\arraystretch}{1.3}
\begin{center}
\begin{tabular}{|c|cccc|}
\hline
Name & $A$ & $C,T,B$ & $x$ & $y$ \\ \hline
EHT-light-A & 1310.7 kB & 99.4 kB & 320 B & 5.1 kB \\ 
EHT-light-B & 2252.8 kB & 119.8 kB & 352 B & 8.8 kB \\ \hline
EHT-medium-A & 2838.5 kB & 227.7 kB & 528 B & 7.4 kB \\ 
EHT-medium-B & 4866.0 kB & 247.7 kB & 528 B & 12.7 kB \\ \hline
EHT-high-A & 4691.5 kB & 311.2 kB & 616 B & 10.5 kB \\ 
EHT-high-B & 7225.3 kB & 354.9 kB & 672 B & 16.1 kB \\ \hline
\end{tabular}
\end{center}
\caption{Sizes  in bytes (B) or  kilobytes (kB).}
\label{table:sizeKeys}
\end{table}

The approximate costs of the primal attack for each parameters' set are reported in Table \ref{table:plaintextRecoveryComplexity}.   The number of samples $m$ and BKZ block size $b$ have been chosen  to minimise  the attack cost. \textbf{C}, \textbf{Q} and \textbf{P} correspond respectively to classical, quantum and  best plausible quantum lattice sieve algorithm complexity according to Section 
\ref{pl_text_recover}.

\begin{table}[h]
\renewcommand{\arraystretch}{1.3}
\begin{center}
\begin{tabular}{|c|ccccc|}
\hline
Name & $m$ & $b$ & \textbf{C} & \textbf{Q} & \textbf{P} \\ \hline
EHT-light-A & 420 & 468 & 136 & 124 & 97 \\ 
EHT-light-B & 471 & 478 & 139 & 126 & 99 \\ \hline
EHT-medium-A & 653 & 728 & 212 & 193 & 151 \\ 
EHT-medium-B & 653 & 728 & 212 & 193 & 151 \\ \hline
EHT-high-A & 788 & 956 & 279 & 253 & 198 \\ 
EHT-high-B & 835 & 947 & 276 & 251 & 196 \\ \hline
\end{tabular}
\end{center}
\caption{Primal attack costs in $\log_2$-scale.  }
\label{table:plaintextRecoveryComplexity}
\end{table}

In Table \ref{table:keyRecovery}  we give the costs for the key-recovery attacks:  brute-force, time-memory-trade-off and the lattice-based attacks according to Section \ref{variation_1}.

%

\begin{table}[h]
\renewcommand{\arraystretch}{1.3}
\begin{center}
\begin{tabular}{|c|ccccc|c|c|}
\hline
\multirow{2}{*}{Name} & \multicolumn{5}{c|}{Lattice-based} & \multirow{2}{*}{Time-Memory} & \multirow{2}{*}{Bruteforce} \\
 & $h$ & $b$ & \textbf{C} & \textbf{Q} & \textbf{P} &  &  \\ \hline
EHT-light-A & 258 & 131 & 309 & 305 & 296 & 305 & 595 \\ 
EHT-light-B & 267 & 127 & 346 & 343 & 334 & 326 & 636 \\ \hline
EHT-medium-A & 314 & 115 & 311 & 308 & 299 & 318 & 620 \\ 
EHT-medium-B & 314 & 115 & 361 & 357 & 349 & 343 & 670 \\ \hline
EHT-high-A & 331 & 108 & 336 & 333 & 324 & 334 & 652 \\ 
EHT-high-B & 344 & 106 & 364 & 361 & 353 & 350 & 684 \\ \hline
\end{tabular}
\end{center}
\caption{Key recovery attacks cost in $\log_2$-scale.} 
\label{table:keyRecovery}
\end{table}

\section{Implementation and Performance}\label{performance} 

 EHT encryption/decryption was implemented in C with no external dependencies. The objective was to have a  proof-of-concept of the EHT, run experiments and make comparisons with other public-key protocols. By no means we claim this implementation to be ready for  the real world, or to be resistant to side-channel attacks, etc. Also, even if the efficiency was one of our main foci, there is certainly room for further optimization, such as the use of AVX2.

We made two implementations to run the same algorithms with a difference in memory usage:  only stack memory and a combination of stack and heap. The former provides with  the best performances and is referred in this section. Both variants of the implementation are publicly available at \url{https://github.com/AlessandroBudroni/EHT-C}.

\subsection{Performance}
To evaluate the performance of EHT we run it for the benchmark parameters in Table \ref{table:proposedParameter}. A machine with processor 3.60GHz Intel Core i7-7700 CPU, running Linux Mint 20 and with 32 GB of RAM was used. As a standard practice, TurboBoost was disabled during the tests. The code was compiled with \texttt{gcc -std=c11 -O3}.

\begin{table}[h]
\renewcommand{\arraystretch}{1.3}
\begin{center}
\begin{tabular}{|c|c|ccc|c|}
\hline
Name & Plaintext Load & KeyGen & Encryption & Decryption & \begin{tabular}[c]{@{}c@{}}Total\\ (Enc + Dec)\end{tabular} \\ \hline
EHT-light-A & 317.5  B &   699188 &       2055 &       8677 &      10732 \\
EHT-light-B & 349.25 B &   735385 &       3114 &      13322 &      16436 \\ \hline
EHT-medium-A & 525.25 B & 2320390 &       3836 &      14739 &      18575 \\
EHT-medium-B & 525.25 B & 2493254 &       6690 &      16843 &      23533 \\ \hline
EHT-high-A & 613.25 B   & 3518341 &       6286 &      29072 &      35358 \\
EHT-high-B & 669.00 B &   3597042 &       8296 &      35192 &      43488 \\ \hline
\end{tabular}
\end{center}
\caption{Performance of EHT in thousands of cycles.  }
\label{table:Performances}
\end{table}

  The performance of the implementation is detailed in Table \ref{table:Performances}. Encryption is considerably faster than decryption and key generation routines. This was  a design choice  when determining the parameters sets. The rationale was to keep the cost low on the client side, i.e. encryption, at the expense of increasing the cost on the server side, i.e. key-generation and decryption.

\subsection{Comparison with FrodoKEM}
 FrodoKEM \cite{FrodoKEM}  is one of the  candidates to the NIST Post-Quantum Standardization Process at Round 3 and the only one purely based on LWE left in the competition.
It is a key encapsulation mechanism designed to allow two parties to agree on a master secret of 16, 24 or 32 bytes. On the other hand, EHT is a public-key crypto-system that allows to encrypt a relatively large number of bytes at once. A scenario on which both the protocols find an application is when two parties must agree on a batch of several keys. Due to the large plaintext load provided by EHT, it is enough to make only one encryption run  to agree on a big master secret that can be used as the source for generating several keys. For example, EHT-light-B allows to exchange 349.25 bytes of plaintext, enough for generating 20 master secrets of 16 bytes each. In this scenario, EHT dominates  over FrodoKEM in speed.

The experimental results  reported in Table \ref{table:FrodoComparison} support that claim. Both the algorithms  were run  the necessary number of times to encrypt and then decrypt (encapsulate and decapsualte)  the same number of plain-text bits for each security level. The implementation of FrodoKEM available at \url{https://github.com/Microsoft/PQCrypto-LWEKE} with AVX2 enabled and \textit{openssl} disabled was used. Plaintext load is expressed in kilobytes (kB), measured time is expressed in seconds (s) and it is the mean of 5 run. One can see that, at all security levels, EHT is considerably faster than FrodoKEM.

\begin{table}[h]
\renewcommand{\arraystretch}{1.3}
\begin{center}
\begin{tabular}{|c|c|c|c|c|}
\hline
Level & Protocol & Plaintext Load & Repetitions & Time \\ \hline
\multirow{4}{*}{1} & EHT-light-A & \multirow{4}{*}{20 kB} & 63 & 0.19 s \\
 & EHT-light-B &  & 58 & 0.26 s \\
 & FrodoKEM-640-AES &  & 1250 & 1.25 s \\
 & FrodoKEM-640-SHAKE &  & 1250 & 3.00 s \\ \hline
\multirow{4}{*}{3} & EHT-medium-A & \multirow{4}{*}{30 kB} & 58 & 0.29 s \\
 & EHT-medium-B &  & 58 & 0.37 s \\
 & FrodoKEM-976-AES &  & 1250 & 2.44 s \\
 & FrodoKEM-976-SHAKE &  & 1250 & 6.35 s \\ \hline
\multirow{4}{*}{5} & EHT-high-A & \multirow{4}{*}{40 kB} & 66 & 0.64 s \\
 & EHT-high-B &  & 60 & 0.71 s \\
 & FrodoKEM-1344-AES &  & 1250 & 4.26 s \\
 & FrodoKEM-1344-SHAKE &  & 1250 & 11.23 s \\ \hline
\end{tabular}
\end{center}
\caption{Performance  of EHT against FrodoKEM.} 
\label{table:FrodoComparison}
\end{table}

\bibliographystyle{alpha}

\begin{thebibliography}{9}

 \bibitem{AD}E. Alkim, L. Ducas, T. P\"{o}ppelmann, P. Schwabe, \emph{Post-quantum key exchange: A new hope}, in USENIX Conference on Security, pp. 327--343. SEC'16, USENIX Association, USA (2016).
 
 \bibitem{FrodoKEM} E. Alkim, J. Bos, L. Ducas, P. Longa, I. Mironov, M. Naehrig, V. Nikolaenko,  C. Peikert, A. Raghunathan, D. Stebila,  \emph{FrodoKEM Learning With Errors Key Encapsulation,} Tech. rep. NIST  (2020), available at \url{https://csrc.nist.gov/projects/post-quantum-cryptography/round-3-submissions}.

\bibitem{BR94} M. Bellare, P. Rogaway, \emph{Optimal Asymmetric Encryption -- How to encrypt with RSA. Extended abstract,} in Eurocrypt '94, LNCS vol. 950, Springer, 1995.   
\bibitem{CT} T. M. Cover and J. A. Thomas,\emph{ Elements of Information Theory}, second edition,  John Wiley \& Sons, 2006.

\bibitem{Hanrot} G. Hanrot, X. Pujol, D. Stehl\'{e},   \emph{Analyzing blockwise lattice algorithms using dynamical systems,} in CRYPTO 2011, LNCS vol. 7073, pp. 1--20, Springer 2011.


\bibitem{HKM17} G. Herod, E. Kirshanova and A. May, \emph{On the asymptotic complexity of solving LWE}, Des. Codes  Crypt. vol. 86(2018), pp. 55--83.

\bibitem{Laarhoven} T. Laarhoven,   \emph{Sieving for shortest vectors in lattices using angular locality-sensitive hashing,} in CRYPTO 2015, LNCS vol. 9215, pp. 3--22, Springer 2015.


\bibitem{Laarhoven15} T. Laarhoven, \emph{ Search problems in cryptography}, available at \url{http://fs.fish.govt.nz/Page.aspx?pk=7\&sc=SUR},  January 2014.



\bibitem{Laarhoven18} T. Laarhoven, A. Mariano, \emph{Progressive lattice sieving}, in  PQCrypto 2018, pp. 292--311, Springer, 2018.

\bibitem{LaarhovenMosca} T. Laarhoven, M. Mosca, J. Pol, \emph{ Finding shortest lattice vectors faster using quantum search}, Des. Codes Crypt. 77(2015), pp.  375--400.

\bibitem{nist} NIST Post-Quantum Cryptography Standardization Process, \url{https://csrc.nist.gov}.
 
\bibitem{Peikert2016} C. Peikert,   \emph{A decade of lattice cryptography,} Foundations and Trends in Theoretical Computer Science, vol. 10, issue 4, pp. 283--424, 2016.




\bibitem{Regev2005} O. Regev,   \emph{On lattices, learning with errors, random linear codes, and cryptography,} Journal of the ACM,  vol. 56(2009), issue 6.

\end{thebibliography}

\subsection{Appendix 1. Asymptotic value  of $k$.}\label{asymptotic_failure}
 If $C_{1+k(i-1)}e,\ldots,C_{k+k(i-1)}e$ are independently distributed, then 
 $
k=O(\ln q/\hbox{Div}(P_{\sigma\lambda}||U))
$
by Chernoff-Stein lemma, see \cite{CT}, where $\hbox{Div}(P_{\sigma\lambda}||U)$ is the Kullback-Leibler divergence of the distribution $ P_{\sigma\lambda}$ and the uniform distribution $U$ on residues modulo $q$. We have
$$
\hbox{Div}(P_{\sigma\lambda}||U)= \sum_{a}p_{\sigma\lambda}[a]\ln(q\,p_{\sigma\lambda}[a])=E\ln(q\,p_{\sigma\lambda}[a]).
$$
The distribution $P_{\sigma\lambda}$ may be approximated by an ordinary normal distribution with mean $0$ and variance $(\sigma\lambda)^2$, see \eqref{approx}. Then $\hbox{Div}(P_{\sigma\lambda}||U)\approx\ln(\frac{q}{\sigma\lambda\sqrt{2\pi}})-\frac{1}{2}=\ln(\frac{q}{\sigma\lambda\sqrt{2\pi e}})$ and  
$
k= O\left(\ln q/\ln(\frac{q}{\sigma\lambda\sqrt{2\pi e}})\right).
$
 If $c_{\sigma}+c_{\lambda}<c_{q}$ and they are constants, then $k$ is bounded by a constant for large $n$.  
   The vector $Ce$ has approximately a multivariate normal distribution with the  covariance matrix $\sigma^2CC^T$. The expected value of non-diagonal entries of $\sigma^2CC^T$ is bounded by a small parameter $\gamma$  by the matrix $C$ construction. When using a Hadamard matrix $H$ for constructing $C$ the non-diagonal entries are $0$. So the statistical dependence of the entries of $Ce$ is negligible and the asymptotic bound on $k$ is  correct.

\end{document}